\documentclass[letterpaper,12pt,reqno]{article}
\usepackage{authblk}
\usepackage{amsthm,amssymb,amsmath}
\usepackage{enumerate}
\usepackage{graphicx,hyperref}
\usepackage{eqnarray,array}
\usepackage{pdfsync}
\usepackage{tikz}
\usepackage{booktabs}
\usepackage{color,xcolor}
\usepackage{tikz,3dplot}
\usepackage{subfig,caption}
\usepackage{chngpage}
\usepackage[top=1in, bottom=1in, left=1in, right=1in]{geometry}
\usepackage[comma,numbers,sort,longnamesfirst]{natbib}
\usetikzlibrary{arrows,positioning}

\newtheorem{thm}{Theorem}
\newtheorem{cor}{Corollary}

\theoremstyle{definition}

\newtheorem{lcon}{Linearity Condition}

\theoremstyle{remark}

\begin{document}

\title{Cluster-Based Regularized Sliced Inverse Regression for Forecasting Macroeconomic Variables}%

%
%

\date{}

\author[1]{Yue Yu}
\author[2]{Zhihong Chen\thanks{Corresponding author. Zhihong Chen acknowledges financial support from the NSF of China (Grant 71101030) and the Program for Innovative Research Team in UIBE(Grant CXTD4-01).}}
\author[3]{Jie Yang}
\affil[1]{TradeLink L.L.C., Chicago, Illinois.} \affil[2]{School of International Trade and
Economics, University of International Business and Economics,
China.}\affil[3]{Department of Mathematics,
Statistics, and Computer Science, University of Illinois at Chicago.}

\renewcommand\Authands{, and }

\maketitle
\clearpage
 \begin{abstract}
 This article concerns the dimension reduction in regression for large data set. We introduce a new method based on the sliced inverse regression approach, called cluster-based regularized sliced inverse regression. Our method not only keeps the merit of considering both response and predictors' information, but also enhances the capability of handling highly correlated variables. It is justified under certain linearity conditions. An empirical application on a macroeconomic data set shows that our method has outperformed the dynamic factor model and other shrinkage methods.
 \end{abstract}


\section{Introduction}

\noindent Forecasting using many predictors has received a good deal of attention in recent years. The curse of dimensionality has been turned into a blessing with the abundant information in large datasets. Various methods have been originated to extract efficient predictors, for example, dynamic factor model (DFM), Bayesian model averaging, Lasso, boosting, etc. Among them, dynamic factor model is conceptually appealing in macroeconomics because it is structurally consistent with log-linearlized models such as dynamic stochastic general equilibrium models.

\citet{boivin2005understanding} assessed the extent to which the forecasts are influenced by how the factors are estimated and/or how the forecasts are formulated. They did not find one method that always stands out to be systematically good or bad. Meta-study from \citet{eickmeier2008successful} also found mixed performance of DFM forecasts. \citet{stock2011generalized} compared the dynamic factor model with some recent multi-predictor methods. They concluded that the dynamic factor model could not be outperformed by these methods for all the forecasting series in their data set.

The recent development in statistics provides a new method of dimension reduction in regression for large-dimensioned data. The literature stems from \citet{duan1991slicing}, and \citet{li1991sliced}, which proposed a new way of thinking in the regression analysis, called sliced inverse regression (SIR). SIR reverses the role of response $y$ and predictors $\mathbf{x}$. Classical regression methods mainly deal with the conditional density ${\bf f}(y|\mathbf{x})$. SIR collects the information of the variation of predictors $\mathbf{x}$ along with the change of the response $y$, by exploring the conditional density $h(\mathbf{x}|y)$. Usually the dimension of the response is far more less than the dimension of the predictors, hence, it is a way to avoid the ``curse of dimensionality''.

The traditional SIR does not work well for highly correlated data, due to the degenerate covariance matrix. This is not feasible when the number of variables $N$ is greater than the number of observations $T$, which happens a lot in economics studies. In addition, the economic variables are often highly correlated or inversely correlated, due to the derivation formula, data sources, and grouping category, for instance, personal consumption expenditures (PCE) and consumer price index (CPI), total employees and unemployment rate, etc. This makes the covariance matrix ill-conditioned, causes the inverse matrix lack of precision and too sensitive to the variation of matrix entries, and leads to unstable results with large standard deviations. There are some extensions of SIR for the highly collinearity data and ``$T < N$'' problems, for example, regularized sliced inverse regression (\citet{zhong2005rsir}, \citet{li2008sliced}) and partial inverse regression (\citet{li2007partial}).

In this article, we propose a new method of dimension reduction, called the cluster-based sliced inverse regression (CRSIR) method, for many predictors in a data rich environment. We evaluate its properties theoretically and use it for forecasting macroeconomic series. Comparison in terms of both in-sample prediction and out-of-sample forecasting simulation shows the advantage of our method.

The remaining of the article is organized as follows. Section 2 introduces cluster-based SIR method with its statistical property. An empirical application on the macroeconomic dataset used by \citet{stock2011generalized} is given in Section 3. Conclusions with some discussions are provided in Section 4.

\section{Modeling and methods}

\noindent The regression model in \citet{li1991sliced} takes the form of
\begin{equation}\label{1}
y=g(\boldsymbol{\beta}_1'\mathbf{x},\boldsymbol{\beta}_2'\mathbf{x},\dots,\boldsymbol{\beta}_K'\mathbf{x},\boldsymbol{\epsilon}),
\end{equation}
where the response $y$ is univariate, $\mathbf{x}$ is an $N$-dimensional vector, and the random error $\boldsymbol{\epsilon}$ is independent of $\mathbf{x}$. Figure \ref{f0} gives a straightforward illustration of Model \eqref{1}, which means that $y$ depends on $\mathbf{x}$ only through the $K$-dimensional subspace spanned by projection vectors $\boldsymbol{\beta}_1, \dots, \boldsymbol{\beta}_K$, known as the effective dimension reducing directions (e.d.r.-directions) (\citet{li1991sliced}).

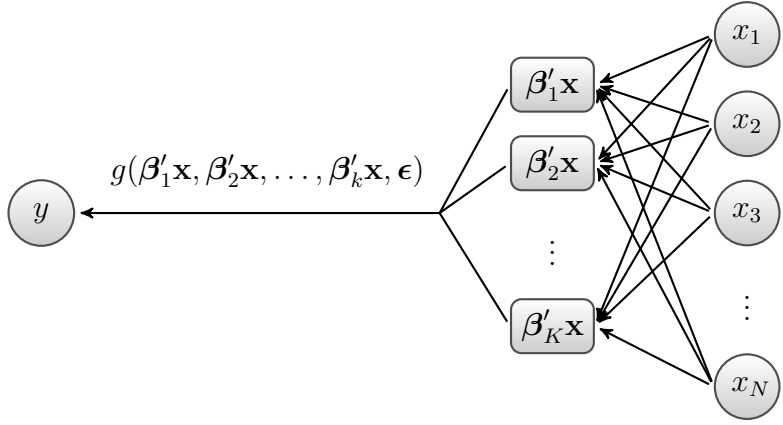
\begin{figure}[!ht]
\begin{center}
\tikzset{
    >=stealth',
    circ/.style={
       circle,
           draw=black!70, thick,
       top color=white,
       bottom color=black!20,
           text width=1em,
           text centered},
    punkt/.style={
       rectangle,
           rounded corners,
       draw=black!70,
       top color=white,
       bottom color=black!20,
           thick,
           text width=2em,
           text centered},
    pil/.style={
           ->,
           thick,
           shorten <=2pt,
           shorten >=2pt,}
}
\begin{tikzpicture}[node distance=1cm, auto,]
 \node[ ](x0){};
 \node[circ, below=0.2cm of x0](x1){$x_1$};
 \node[circ, below=0.3cm of x1](x2){$x_2$};
 \node[circ, below=0.3cm of x2](x3){$x_3$};
 \node[below=0.3cm of x3](x4){$\vdots$};
 \node[circ, below=0.3cm of x4](x5){$x_N$};
 \node[left=2cm of x1](bx0){};
 \node[punkt, below=0.15cm of bx0](bx1){$\boldsymbol{\beta}_1'\mathbf{x}$};
 \node[punkt, below=0.3cm of bx1](bx2){$\boldsymbol{\beta}_2'\mathbf{x}$};
 \node[below=0.3cm of bx2](bx3){$\vdots$};
 \node[punkt, below=0.3cm of bx3](bx4){$\boldsymbol{\beta}_K'\mathbf{x}$};
 \node[left=3.5cm of x3](nul1){};
 \node[circ,left=4.7cm of nul1](y){$y$};
 \node[left=3.5cm of bx0](nul2){};
 \node[below=1.3cm of nul2](form){$g(\boldsymbol{\beta}_1'\mathbf{x},\boldsymbol{\beta}_2'\mathbf{x},\dots,\boldsymbol{\beta}_k'\mathbf{x},\boldsymbol{\epsilon})$};

 \path (x1.west) edge[pil] (bx1.east);
 \path (x1.west) edge[pil] (bx2.east);
 \path (x1.west) edge[pil] (bx4.east);
 \path (x2.west) edge[pil] (bx1.east);
 \path (x2.west) edge[pil] (bx2.east);
 \path (x2.west) edge[pil] (bx4.east);
 \path (x3.west) edge[pil] (bx1.east);
 \path (x3.west) edge[pil] (bx2.east);
 \path (x3.west) edge[pil] (bx4.east);
 \path (x5.west) edge[pil] (bx1.east);
 \path (x5.west) edge[pil] (bx2.east);
 \path (x5.west) edge[pil] (bx4.east);
 \path (bx1.west) edge[-,thick,shorten <= 2pt] (nul1.center);
 \path (bx2.west) edge[-,thick,shorten <= 2pt] (nul1.center);
 \path (bx4.west) edge[-,thick,shorten <= 2pt] (nul1.center);
 \path (nul1.center) edge[pil, shorten <=0pt] (y.east);

\end{tikzpicture}
\caption{Regression Model \eqref{1} Using E.D.R.-Directions}\label{f0}
\end{center}
\end{figure}

Many methods can be used to find the e.d.r.-directions, for example, principal component analysis might be the most commonly used one in economics. But unlike these methods, SIR not only reduces dimensions in regression but also integrates the information from both predictors and response. Moreover, different from the classical regression methods, SIR intends to collect information on how $\mathbf{x}$ changes along with $y$. That is to say, instead of estimating the forward regression function ${\boldsymbol\eta}(\mathbf{x})=E(y|\mathbf{x})$, inverse regression considers ${\boldsymbol\xi}(y)=E(\mathbf{x}|y)$. Compared with ${\boldsymbol\eta}(\mathbf{x})$, the inverse regression function $\boldsymbol\xi(y)$ depends on one-dimensioned $y$, which makes the operation much easier.

\citet{li1991sliced} showed that using SIR method, the e.d.r.-directions can be estimated by solving
\begin{equation}\label{2}
\mathrm{Cov}\big(E(\mathbf{x}|y)\big)\boldsymbol{\beta}_j=\nu_j \mathrm{Cov}(\mathbf{x}) \boldsymbol{\beta}_j,
\end{equation}
where $\nu_j$ is the $j$th eigenvalue and $\boldsymbol{\beta}_j$ is the corresponding eigenvector of $\mathrm{Cov}\big(E(\mathbf{x}|y)\big)$ with respect to $\mathrm{Cov}(\mathbf{x})$. During the forecasting procedure, the covariance matrices can be replaced by their usual moment estimates.

One of the key parameters used in SIR is the number of slices $H$. However, it is not crucial for larger sample sizes, \citet{li2000high} indicated that for a sample size $n=300$, $H$ can be chosen between 10 to 20, and SIR outputs do not change much for a wide range of $H$. Accordingly, we fix $H=10$ throughout this article.

\subsection{Cluster-based sliced inverse regression}

In this section, we introduce clustering methodology with the sliced inverse regression to improve the performance of SIR on collinear data.

Assuming that the variables of interest can be clustered into
several blocks, so that two variables within the same block are
correlated to each other, and any two variables belonging to
different blocks are independent. In practice, an orthogonalization
procedure can be applied to reduce the correlations between blocks
in order to fit our assumption. Thus, we can cluster the variables
according to their correlations in order to find the
e.d.r-directions, because there is no shared information between
clusters.

The clustering method we use is hierarchical clustering (\citet{ward1963hierarchical}) with complete linkage. The dissimilarity is defined as $1-|\textrm{Correlation}|$.

The algorithm for the cluster-based SIR method can be described as following.
\begin{enumerate}
  \item Standardize each explanatory variable to zero mean and unit variance.
  \item Cluster $\mathbf{x}\ (N\times 1)$ into $\big(\begin{array}{ccc}\mathbf{x_1}' & \cdots & \mathbf{x_c}'\end{array}\big)'$ based on the correlation matrix of $\mathbf{x}$, where $\mathbf{x_i}$ is $N_i \times 1$, $\sum_{i=1}^{c}N_i=N$, and $c$ is the number of clusters, which will be determined by cross-validation.
  \item Restricted to each cluster, perform SIR method and pick up $k_i$ SIR directions based on the sequential chi-square test (Li, 1991), say $\boldsymbol{\theta}_j^{(i)},\ j=1,\dots,k_i,\ i=1,\dots,c$.
  \item Collect all the SIR variates obtained from the clusters, say $\{\boldsymbol{\theta}_j^{(i)\prime}\mathbf{x}_i, i=1,2,\dots,c, j=1,2,\dots,k_i\}$.
  \item Let $\boldsymbol{\lambda}_l=\big(\begin{array}{ccc}\mathbf{0_1}' &\boldsymbol{\theta}_j^{(i)\prime}& \mathbf{0_2}'\end{array}\big)'$, $l=1,\dots,m$, $m=\sum_{i=1}^{c}k_i$, where $\mathbf{0_1}$ and $\mathbf{0_2}$ are zero column vectors with dimension $\sum_{k=1}^{i-1}N_k$ and $\sum_{k=i+1}^{c}N_k$, respectively. Denote $\Lambda=(\boldsymbol{\lambda}_1,\boldsymbol{\lambda}_2,\dots, \boldsymbol{\lambda}_m)$. The variates $\{\boldsymbol{\theta}_j^{(i)\prime}\mathbf{x}_i\}$ can be written in a vector form as $(\boldsymbol{\lambda}_1'\mathbf{x},\dots,\boldsymbol{\lambda}_m'\mathbf{x})'=\Lambda'\mathbf{x}$.
  \item Then, perform SIR method one more time to the pooled variates $\Lambda'\mathbf{x}$ to reduce dimensions further, and get the e.d.r.-directions $(\boldsymbol{\gamma}_1, \boldsymbol{\gamma}_2,\dots,\boldsymbol{\gamma}_v)$, where $v$ is also determined by the sequential chi-square test. Denote $\Gamma=(\boldsymbol{\gamma}_1, \boldsymbol{\gamma}_2,\dots,\boldsymbol{\gamma}_v)$, the final CRSIR variates we chose are $\Gamma'\Lambda'\mathbf{x}$.
  \item Estimate the values of forecasting series using the CRSIR variates $\Gamma'\Lambda'\mathbf{x}$. Linear model with ordinary least squares (OLS) is used in this article, and as to be shown later, it is sufficiently good for our method.
\end{enumerate}
Note that the matrices $\Gamma$ is $m \times v$, $\Lambda$ is $N \times m$, so $\Gamma'\Lambda'\mathbf{x}$ is $v \times 1$. Therefore, we only use $v$ factors to build the final model for forecasting $y$, instead of using $N$ variables based on the original dataset.

\subsection{Statistical property of cluster-based SIR}

\citet{li1991sliced} established the unbiasedness for the e.d.r.-directions found by SIR, assuming the following linearity condition.
\begin{lcon}\label{con0}
For any $\mathbf{b} \in \mathbb{R}^N$, the conditional expectation $E(\mathbf{b'x}|\allowbreak \boldsymbol{\beta}_1'\mathbf{x},\dots, \boldsymbol{\beta}_K'\mathbf{x})$ is linear in $\boldsymbol{\beta}_1'\mathbf{x},\dots,\boldsymbol{\beta}_K'\mathbf{x}$.
\end{lcon}
The linearity condition is not easy to verify, however, \citet{eaton1986characterization} showed when $\mathbf{x}$ is elliptically symmetrically distributed, for example, multivariate normally distributed, the linearity condition holds. Furthermore, \citet{hall1993almost} showed that elliptical symmetric distribution is not a restrictive assumption, because the linearity condition holds approximately when $N$ is large even if the dataset has not been generated from an elliptically symmetric distribution.

Without loss of generality, we assume each variable in ${\bf x}$ has been standardized to zero mean and unit variance for our discussion. \citet{li1991sliced} proved the following theorem,
\begin{thm}\label{thm0}
Assume Linearity Condition \ref{con0}, the centered inverse regression curve $E(\mathbf{x}|y)$ is contained in the space spanned by $\Sigma_{\mathbf{x}}\boldsymbol{\beta}_j,\ j=1,\dots,K$, where $\Sigma_{\mathbf{x}}$ is the covariance matrix of $\mathbf{x}$.
\end{thm}

Figure \ref{f0.5} shows a three-dimensional case when $\mathbf{x}=\big(x_1,x_2,x_3\big)'$, since the inverse regression function $E(\mathbf{x}|y)$ is a function of $y$, it draws a curve in the three-dimensional space when $y$ changes. Theorem \ref{thm0} indicates that such curve is located exactly on the plane spanned by two directions $d_1$ and $d_2$ from $\Sigma_{\mathbf{x}}\boldsymbol{\beta}_j,\ j=1,2$,  assuming $K=2$.

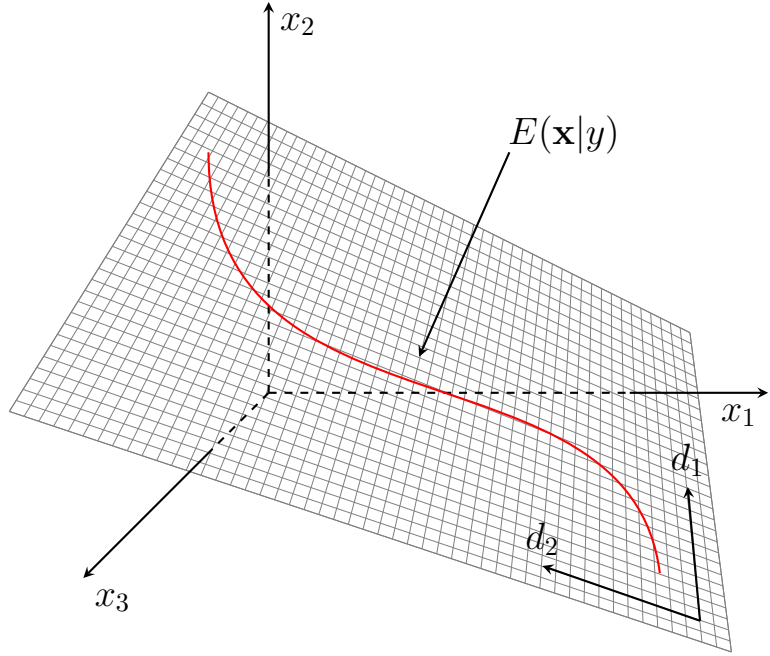
\begin{figure}[!tpb]
\begin{center}
\tdplotsetmaincoords{60}{110}
\pgfmathsetmacro{\rvec}{.8}
\pgfmathsetmacro{\thetavec}{30}
\pgfmathsetmacro{\phivec}{60}
\begin{tikzpicture}[scale=0.8]
\tikzstyle{every node}=[font=\large]
\coordinate (O) at (0,0,0);
\tdplotsetcoord{P}{\rvec}{\thetavec}{\phivec}
\draw[gray, thin,-] (-4,5,0) -- (4,1,0);
\draw[gray, thin, -] (-5,2,6) -- (7,-2,6);
\draw[gray, thin, -] (-4,5,0) -- (-5,2,6);
\draw[gray, thin, -] (7,-2,6) -- (4,1,0);
\draw[gray, very thin, -] (-4,5,0) -- (-5,2,6);
\draw[gray, very thin, -] (-3.8367,4.9184,0) -- (-4.7551,1.9184,6);
\draw[gray, very thin, -] (-3.6735,4.8367,0) -- (-4.5102,1.8367,6);
\draw[gray, very thin, -] (-3.5102,4.7551,0) -- (-4.2653,1.7551,6);
\draw[gray, very thin, -] (-3.3469,4.6735,0) -- (-4.0204,1.6735,6);
\draw[gray, very thin, -] (-3.1837,4.5918,0) -- (-3.7755,1.5918,6);
\draw[gray, very thin, -] (-3.0204,4.5102,0) -- (-3.5306,1.5102,6);
\draw[gray, very thin, -] (-2.8571,4.4286,0) -- (-3.2857,1.4286,6);
\draw[gray, very thin, -] (-2.6939,4.3469,0) -- (-3.0408,1.3469,6);
\draw[gray, very thin, -] (-2.5306,4.2653,0) -- (-2.7959,1.2653,6);
\draw[gray, very thin, -] (-2.3673,4.1837,0) -- (-2.551,1.1837,6);
\draw[gray, very thin, -] (-2.2041,4.102,0) -- (-2.3061,1.102,6);
\draw[gray, very thin, -] (-2.0408,4.0204,0) -- (-2.0612,1.0204,6);
\draw[gray, very thin, -] (-1.8776,3.9388,0) -- (-1.8163,0.9388,6);
\draw[gray, very thin, -] (-1.7143,3.8571,0) -- (-1.5714,0.8571,6);
\draw[gray, very thin, -] (-1.551,3.7755,0) -- (-1.3265,0.7755,6);
\draw[gray, very thin, -] (-1.3878,3.6939,0) -- (-1.0816,0.6939,6);
\draw[gray, very thin, -] (-1.2245,3.6122,0) -- (-0.8367,0.6122,6);
\draw[gray, very thin, -] (-1.0612,3.5306,0) -- (-0.5918,0.5306,6);
\draw[gray, very thin, -] (-0.898,3.449,0) -- (-0.3469,0.449,6);
\draw[gray, very thin, -] (-0.7347,3.3673,0) -- (-0.102,0.3673,6);
\draw[gray, very thin, -] (-0.5714,3.2857,0) -- (0.1429,0.2857,6);
\draw[gray, very thin, -] (-0.4082,3.2041,0) -- (0.3878,0.2041,6);
\draw[gray, very thin, -] (-0.2449,3.1224,0) -- (0.6327,0.1224,6);
\draw[gray, very thin, -] (-0.0816,3.0408,0) -- (0.8776,0.0408,6);
\draw[gray, very thin, -] (0.0816,2.9592,0) -- (1.1224,-0.0408,6);
\draw[gray, very thin, -] (0.2449,2.8776,0) -- (1.3673,-0.1224,6);
\draw[gray, very thin, -] (0.4082,2.7959,0) -- (1.6122,-0.2041,6);
\draw[gray, very thin, -] (0.5714,2.7143,0) -- (1.8571,-0.2857,6);
\draw[gray, very thin, -] (0.7347,2.6327,0) -- (2.102,-0.3673,6);
\draw[gray, very thin, -] (0.898,2.551,0) -- (2.3469,-0.449,6);
\draw[gray, very thin, -] (1.0612,2.4694,0) -- (2.5918,-0.5306,6);
\draw[gray, very thin, -] (1.2245,2.3878,0) -- (2.8367,-0.6122,6);
\draw[gray, very thin, -] (1.3878,2.3061,0) -- (3.0816,-0.6939,6);
\draw[gray, very thin, -] (1.551,2.2245,0) -- (3.3265,-0.7755,6);
\draw[gray, very thin, -] (1.7143,2.1429,0) -- (3.5714,-0.8571,6);
\draw[gray, very thin, -] (1.8776,2.0612,0) -- (3.8163,-0.9388,6);
\draw[gray, very thin, -] (2.0408,1.9796,0) -- (4.0612,-1.0204,6);
\draw[gray, very thin, -] (2.2041,1.898,0) -- (4.3061,-1.102,6);
\draw[gray, very thin, -] (2.3673,1.8163,0) -- (4.551,-1.1837,6);
\draw[gray, very thin, -] (2.5306,1.7347,0) -- (4.7959,-1.2653,6);
\draw[gray, very thin, -] (2.6939,1.6531,0) -- (5.0408,-1.3469,6);
\draw[gray, very thin, -] (2.8571,1.5714,0) -- (5.2857,-1.4286,6);
\draw[gray, very thin, -] (3.0204,1.4898,0) -- (5.5306,-1.5102,6);
\draw[gray, very thin, -] (3.1837,1.4082,0) -- (5.7755,-1.5918,6);
\draw[gray, very thin, -] (3.3469,1.3265,0) -- (6.0204,-1.6735,6);
\draw[gray, very thin, -] (3.5102,1.2449,0) -- (6.2653,-1.7551,6);
\draw[gray, very thin, -] (3.6735,1.1633,0) -- (6.5102,-1.8367,6);
\draw[gray, very thin, -] (3.8367,1.0816,0) -- (6.7551,-1.9184,6);
\draw[gray, very thin, -] (4,1,0) -- (7,-2,6);

\draw[gray, very thin, -] (-4,5,0) -- (4,1,0);
\draw[gray, very thin, -] (-4.0345,4.8966,0.2069) -- (4.1034,0.8966,0.2069);
\draw[gray, very thin, -] (-4.069,4.7931,0.4138) -- (4.2069,0.7931,0.4138);
\draw[gray, very thin, -] (-4.1034,4.6897,0.6207) -- (4.3103,0.6897,0.6207);
\draw[gray, very thin, -] (-4.1379,4.5862,0.8276) -- (4.4138,0.5862,0.8276);
\draw[gray, very thin, -] (-4.1724,4.4828,1.0345) -- (4.5172,0.4828,1.0345);
\draw[gray, very thin, -] (-4.2069,4.3793,1.2414) -- (4.6207,0.3793,1.2414);
\draw[gray, very thin, -] (-4.2414,4.2759,1.4483) -- (4.7241,0.2759,1.4483);
\draw[gray, very thin, -] (-4.2759,4.1724,1.6552) -- (4.8276,0.1724,1.6552);
\draw[gray, very thin, -] (-4.3103,4.069,1.8621) -- (4.931,0.069,1.8621);
\draw[gray, very thin, -] (-4.3448,3.9655,2.069) -- (5.0345,-0.0345,2.069);
\draw[gray, very thin, -] (-4.3793,3.8621,2.2759) -- (5.1379,-0.1379,2.2759);
\draw[gray, very thin, -] (-4.4138,3.7586,2.4828) -- (5.2414,-0.2414,2.4828);
\draw[gray, very thin, -] (-4.4483,3.6552,2.6897) -- (5.3448,-0.3448,2.6897);
\draw[gray, very thin, -] (-4.4828,3.5517,2.8966) -- (5.4483,-0.4483,2.8966);
\draw[gray, very thin, -] (-4.5172,3.4483,3.1034) -- (5.5517,-0.5517,3.1034);
\draw[gray, very thin, -] (-4.5517,3.3448,3.3103) -- (5.6552,-0.6552,3.3103);
\draw[gray, very thin, -] (-4.5862,3.2414,3.5172) -- (5.7586,-0.7586,3.5172);
\draw[gray, very thin, -] (-4.6207,3.1379,3.7241) -- (5.8621,-0.8621,3.7241);
\draw[gray, very thin, -] (-4.6552,3.0345,3.931) -- (5.9655,-0.9655,3.931);
\draw[gray, very thin, -] (-4.6897,2.931,4.1379) -- (6.069,-1.069,4.1379);
\draw[gray, very thin, -] (-4.7241,2.8276,4.3448) -- (6.1724,-1.1724,4.3448);
\draw[gray, very thin, -] (-4.7586,2.7241,4.5517) -- (6.2759,-1.2759,4.5517);
\draw[gray, very thin, -] (-4.7931,2.6207,4.7586) -- (6.3793,-1.3793,4.7586);
\draw[gray, very thin, -] (-4.8276,2.5172,4.9655) -- (6.4828,-1.4828,4.9655);
\draw[gray, very thin, -] (-4.8621,2.4138,5.1724) -- (6.5862,-1.5862,5.1724);
\draw[gray, very thin, -] (-4.8966,2.3103,5.3793) -- (6.6897,-1.6897,5.3793);
\draw[gray, very thin, -] (-4.931,2.2069,5.5862) -- (6.7931,-1.7931,5.5862);
\draw[gray, very thin, -] (-4.9655,2.1034,5.7931) -- (6.8966,-1.8966,5.7931);
\draw[gray, very thin, -] (-5,2,6) -- (7,-2,6);

\draw[dashed,thick,-] (-3,0,0) -- (3,0,0);
\draw[thick,-stealth] (3,0,0) -- (5.3,0,0) node[anchor=north east]{$x_1$};
\draw[dashed,thick,-] (-3,0,0) -- (-3,3.6,0);
\draw[thick,-stealth] (-3,3.6,0) -- (-3,6.5,0) node[anchor=north west]{$x_2$};
\draw[dashed,thick,-] (-3,0,0) -- (-3,0,2.42);
\draw[thick,-stealth] (-3,0,2.42) -- (-3,0,8) node[anchor=north west]{$x_3$};

\draw[thick,-stealth] (6.17,-1.79,5.20) -- (5.08,-0.45,2.90) node[anchor=south]{$d_1$};
\draw[thick,-stealth] (6.17,-1.79,5.20) -- (1.55,-2.88,0) node[anchor=south]{$d_2$};

\draw [thick,red] (-4,4,0) .. controls (-4,-1,0) and (3,1,0) .. (3.5,-3,0);
\draw[thick,-stealth] (1,4,0) -- (-0.5,0.6,0);
\node[right] at (0.8,4.3,0) {$E(\mathbf{x}|y)$};
\end{tikzpicture}
\caption{Inverse Regression Curve in a Three-Dimensional Space}\label{f0.5}
\end{center}
\end{figure}

Similar unbiasedness property can be proved for our cluster-based SIR.

\begin{thm}\label{thm3}
Under certain linearity conditions, $E(\mathbf{x}|y)$ is contained in the space spanned by $\Sigma_{\mathbf{x}}\Lambda \Gamma$.
\end{thm}

Theorem \ref{thm3} describes the desirable property that there is no estimation bias. The e.d.r.-space estimated by our CRSIR method contains the true inverse regression curve. The details of the proof are provided in the Appendix.

\subsection{Orthogonalization}

For a given dataset $\mathbf{X}$ with dimension $N \times T$, and clusters $\mathbf{X_1},\dots,\mathbf{X_c}$, the correlations between these clusters need to be reduced to zero, to achieve cluster-wise independence. QR decomposition along with projection operators is used to perform the orthogonalization.

To begin with, use QR decomposition to find the orthogonal bases of the first cluster $\mathbf{X_1}$, named as $\mathbf{Q_1}$. Next, project the second cluster $\mathbf{X_2}$ onto the space of $\mathrm{span}\{\mathbf{Q_1}\}^{\bot}$, which is the orthogonal complement of the space spanned by $\mathbf{X_1}$, named as $\mathbf{X_2}^*$,
\begin{equation}
\mathbf{X_2}^*=(\mathbf{I}-\mathbf{Q_1}\mathbf{Q_1}')\mathbf{X_2}.
\end{equation}

Then use QR decomposition again to find the orthogonal bases of $\mathbf{X_2}^*$, named as $\mathbf{Q_2}$, and project $\mathbf{X_3}$ onto the space of $\mathrm{span}\{\mathbf{Q_1}, \mathbf{Q_2}\}^{\bot}$, named as $\mathbf{X_3}^*$. Keep doing such process till the last cluster $\mathbf{X_c}$, we will get a new sequence of clusters $\mathbf{X_1}, \mathbf{X_2}^*, \dots,\mathbf{X_c}^*$, in which every two clusters are orthogonal, and the new sequence contains all the information of the original dataset $\mathbf{X}$.

\subsection{Regularization}

Due to the high correlations between the series within each cluster, the covariance matrices of each cluster $\Sigma_{\mathbf{x_i}}$ are ill-conditioned, which make them hard to be inversed. We suggest a regularized version of the covariance matrix to overcome this issue (\citet{friedman1989regularized}).
\begin{equation}
\Sigma_{\mathbf{x_i}}(\tau)=(1-\tau)\Sigma_{\mathbf{x_i}}+\tau\frac{\mathrm{tr}{\Sigma_{\mathbf{x_i}}}}{N_i}I_{N_i},
\end{equation}
where $\tau  \in [0,1]$ is the shrinkage parameter. This is similar to the ridge version proposed by \citet{zhong2005rsir}, which replaces $\Sigma_{\mathbf{x_i}}$ with $\Sigma_{\mathbf{x_i}}+\tau I_{N_i}$.

The shrinkage parameter $\tau$ can be chosen by cross-validation. Note when $\tau =1$, the regularized covariance matrix will degenerate to a diagonal matrix whose diagonal elements are the means of the eigenvalues of $\Sigma_{\mathbf{x_i}}$. In such case, the chosen e.d.r.-direction is one of the input series, and the other series, which may also contain information for the predictors, are discarded.

\subsection{Comparison between CRSIR and SIR}

Before applying the proposed CRSIR method to real data, consider the following simulated example first, for comparing the performance of CRSIR and SIR methods.

We choose $\gamma$ clusters of predictors with cluster size 10, say, $\mathbf{x_1},\ \mathbf{x_2},\cdots,\mathbf{x_{\gamma}}$, which are independent and identically distributed (i.i.d.) with multivariate normal distribution $N(0,\Sigma)$, where $\Sigma$ is a $10 \times 10$ covariance matrix with $1$ at diagonal and $0.9$ at off-diagonal.

The response $y$ is simulated using the following formula,
\begin{equation*}
y=\sum_{j=1}^{\gamma}j\times\mathbf{x_j}+\mathbf{e}
\end{equation*}

Where the random error $\mathbf{e}$ is independent to $\mathbf{x_i}$'s, and follows normal distribution $N(0,0.1)$.

For simplicity, as well as keeping consistent with our following example, root mean square error (RMSE) is considered as a criterion to evaluate both in-sample prediction and out-of-sample forecasting.
\begin{equation}
\mathrm{RMSE}=\sqrt{\sum_{i=1}^T\big(\hat{y}_i-y_i\big)^2\Big/T},
\end{equation}
where $\hat{y}_i$ is the $i$th predicted value of the response, $y_i$ is the $i$th observed value, and $T$ is the number of observations.

We simulate 600 observations, in which 300 of them are used as training data and the others are used as testing data, at each run under above conditions. In CRSIR, the parameters $c$ and $\tau$ are chosen to minimize the in-sample RMSE for each run. Table \ref{t0} presents the means and standard deviations (in the parentheses) for the RMSE of SIR and CRSIR across 100 runs for several cluster numbers, and the median of the corresponding optimal $c$ and $\tau$.

\begin{table}[b!]
\begin{adjustwidth}{-1in}{-1in}
\begin{center}
\caption{Simulation Results for CRSIR and SIR}
\begin{footnotesize}
\begin{tabular}{l|rr|rrrr}
\toprule
 &  \multicolumn{2}{c|}{{\bf SIR}} & \multicolumn{4}{c}{{\bf CRSIR}}\\
 &  In-sample & Out-of-sample & In-sample & Out-of-sample & \emph{median}($c$) & \emph{median}($\tau$)\\
 \hline
 $\gamma=1$ & 1.64(0.40) & 1.65(0.41) & 1.63(0.29) & 1.65(0.35) & 2 & 0.52 \\
 $\gamma=2$ & 3.25(0.63) & 3.41(0.69) & 3.09(0.42) & 3.22(0.46) & 2 & 0.56\\
 $\gamma=5$ & 8.24(1.06) & 9.07(1.34) & 6.34(0.46) & 6.56(0.45) & 5 & 0.71\\
 $\gamma=10$ & 16.62(2.01) & 20.41(2.61) & 10.08(0.76) & 10.55(0.94) & 10 & 0.74\\
\bottomrule
\end{tabular}
\label{t0}
\end{footnotesize}
\end{center}
\end{adjustwidth}
\end{table}

From Table \ref{t0}, it is clear that CRSIR has better results than SIR for both in-sample prediction and out-of-sample forecasting. The CRSIR has similar RMSEs with smaller standard deviations when the number of clusters degenerates to 1, but it appears to be superior when the number of clusters increases. Besides, CRSIR outstands itself in out-of-sample forecasting, RMSEs for the testing data are almost the same as the one for training data, while SIR has much larger out-of-sample RMSEs. In addition, our other simulations, which are not presented here, show that CRSIR performs even better when the sample size $T$ decreases to $N$.

\section{Empirical application}

\subsection{Dataset and method}

The dataset we use is \citet{stock2011generalized} dataset, which contains 143 quarterly macroeconomic variables from 13 economic categories, such as gross domestic product (GDP), industrial production (IP), employment, price indexes, interest rates, etc. We use 109 of them as explanatory variables, since the other 34 are just high-level aggregates of the 109. All 143 variables are used for forecasting purposes.

Following Stock and Watson's data transformation methodology, first differences of logarithms, first differences, and second differences of logarithms are used for real quantity variables, nominal interest rate variables, and price series, respectively. The correlation plot of the 109 predictor series after logarithm and/or differences is showed in Figure \ref{f1}, which demonstrates that there do exist some highly correlated blocks. Therefore, our cluster-based method is necessary for this dataset.

\begin{figure}[p!]
  \begin{center}
  \includegraphics[width=0.9\textwidth]{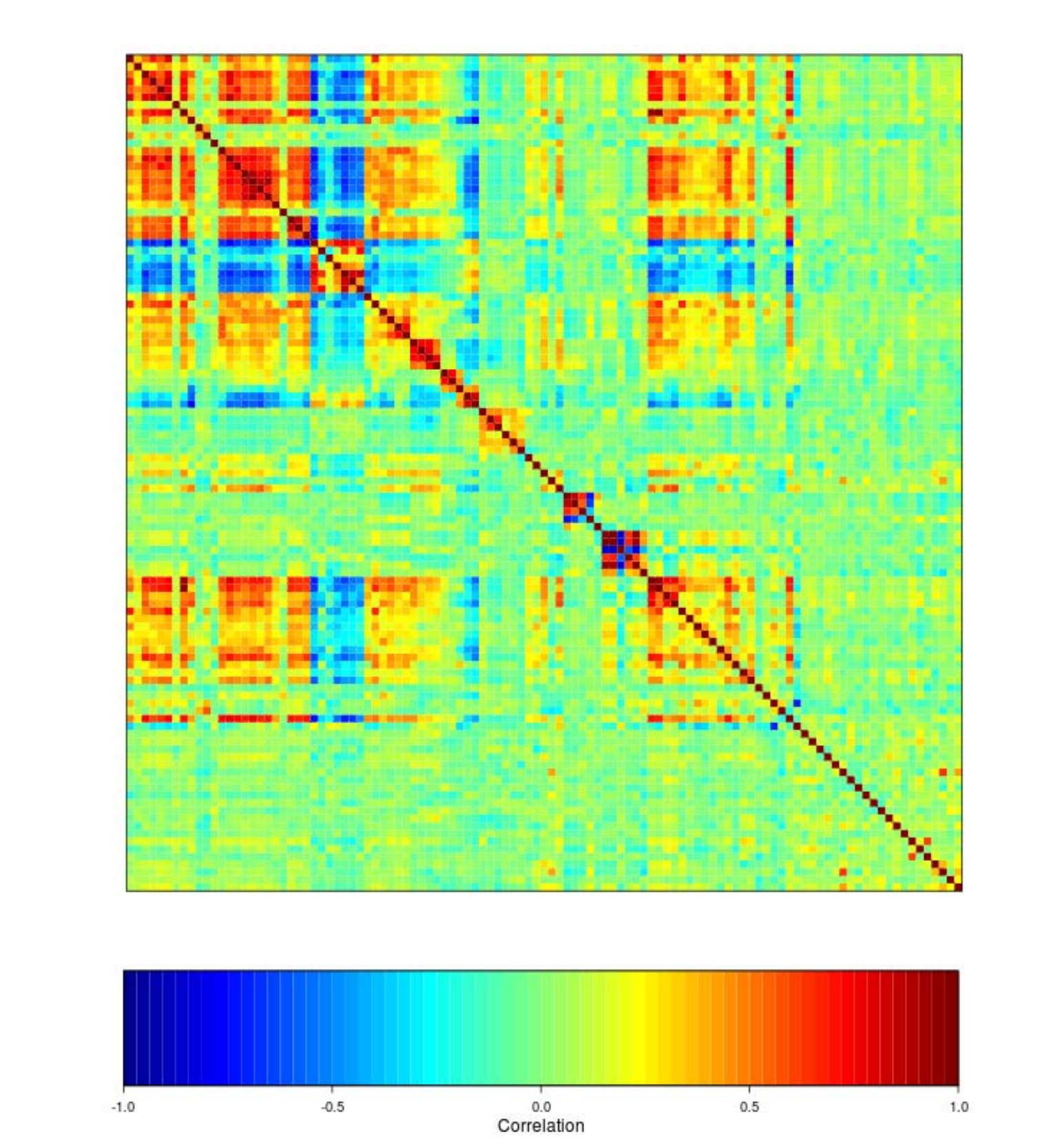}
  \caption{Plot of Correlations of the 109 Predictor Series}
\label{f1}
\end{center}
\end{figure}

For the purpose of comparison, similar rolling pseudo out-of-sample forecasting simulation as in \citet{stock2011generalized} is used, as well as the cross validation for choosing $c$ and $\tau$. In general, starting from 1985 to 2008, at each date $t$, using the data prior to $t$ to predict the forecasted variable $y$ at $h$ date ahead, which is denoted as $\hat{y}_{t+h}$. The main steps can be described as follows,
\begin{enumerate}
 \item Use the formula given by \citet{stock2011generalized}  Table B.2 to transform all the series and screen for outliers.
 \item At each date $t$, use cross-validation, which is described below, to the most recent 100 observations to choose the parameter $c$ and $\tau$ in CRSIR based on mean square error.
 \item Use the chosen $\hat{c}$ and $\hat{\tau}$ with the data prior to $t$ to predict $\hat{y}_{t+h}$ by CRSIR.
 \item Calculate the RMSE for the forecasting procedure,
$$\textrm{RMSE}=\sqrt{\sum_{t=1}^{T}\big(y_{t+h}-\hat{y}_{t+h}\big)^2/T}.$$
\end{enumerate}
The steps for cross-validation are described as follows,
\begin{enumerate}[(i)]
 \item Regress $y_{t+h}$ and $x_t$ on the autoregressive terms $1$, $y_{t}$, $y_{t-1}$, $y_{t-2}$, $y_{t-3}$, in order to eliminate the autoregressive effect. Denote the residuals as $\tilde{y}_{t+h}$ and $\tilde{x}_{t}$.
 \item Let $\Im(t)=\{1,\cdots,t-2h-3,t+2h+3,\cdots,100\}$, at each date $t=1,\cdots,100-h$, find the e.d.r-directions and linear regression model using CRSIR and observations $\tilde{y}_{i}$ and $\tilde{x}_{i},\ i \in \Im(t)$.
 \item Use the e.d.r-directions and linear regression model from the above step at date $t$ to predict $\tilde{y}_{t+h}$.
 \item For fixed $h$, parameters $(c,\tau)$ are chosen by minimizing the sum of squared forecasting error,
$$(\hat{c},\hat{\tau})=\textrm{argmin} \frac{1}{100-h}\sum_{t=1}^{100-h}\Big(\tilde{y}_{t+h}-\hat{\tilde{y}}_{t+h}\Big)^2.$$
\end{enumerate}

\subsection{Results}

We compare our method with the dynamic factor model using the first five principle components (DFM-5), which was claimed to be no worse than any other shrinkage methods according to Stock and Watson (2011). Besides, autoregressive model of order 4 (AR(4)) is used as a benchmark, and all RMSEs are recorded as the ratio relative to AR(4), smaller relative RMSE indicates better forecasting performance.

Table \ref{t1} presents the number of series with smaller RMSEs than AR(4) model for CRSIR and DFM-5. We can see that for forecasting period $h=1$, if CRSIR is used, there are 97 series out of 143 have smaller RMSEs than the benchmark AR(4) model. If DFM-5 is used, only 85 series out of 143 have smaller RMSEs than AR(4) model. The differences become even larger for big forecasting period, when $h=4$ the number of series of CRSIR increases to 115 while the number of DFM-5 decreases to 53.

\begin{table}[!htbp]
\begin{center}
\caption{Number of Series with Smaller RMSE than AR(4) Model}
\begin{tabular}{c|cc}
\toprule
 & {\bf DFM-5} & {\bf CRSIR} \\ \hline
$h=1$ & 85 & 97 \\
$h=2$ & 59 & 109 \\
$h=4$ & 53 & 115 \\
\bottomrule
\end{tabular}
\label{t1}
\end{center}
\end{table}

Table \ref{t2} presents the distributions of the RMSEs for AR(4), DFM-5, and CRSIR methods. When $h=1$, the first quartile of the relative RMSE of CRSIR is just 0.768, which is much smaller than the relative RMSE of DFM-5 (0.961), and the median relative RMSE of CRSIR is 0.907, while DFM-5 has 0.993. When $h=2$ and $h=4$, CRSIR improves the forecasting results of AR(4) for more than $3/4$ of the series. The relative RMSEs of CRSIR at first, second, and third quartile are all smaller than those of DFM-5.

From Table \ref{t1} and \ref{t2}, one can tell that CRSIR improves the forecasting results significantly compared to the DFM-5 method, especially for longer forecasting period.

\begin{table}[!htbp]
\begin{center}
\caption{Distributions of Relative RMSEs by Pseudo Out-of-Sample Forecasting}
\label{t2}
\subfloat[{\normalsize $h=1$}\label{t2a}]{
\begin{tabular}{l|rrrrr}
\toprule
{\bf Method} & \multicolumn{5}{c}{{\bf Percentiles}} \\
 & 0.050 & 0.250 & 0.500 & 0.750 & 0.950 \\ \hline
AR(4) & 1.000 & 1.000 & 1.000 & 1.000 & 1.000 \\
DFM-5 & 0.874 & 0.961 & 0.993 & 1.022 & 1.089 \\
CRSIR & 0.621 & 0.768 & 0.907 & 1.048 & 1.372 \\
\bottomrule
\end{tabular}
}

\subfloat[{\normalsize $h=2$}\label{t2b}]{
\begin{tabular}{l|rrrrr}
\toprule
{\bf Method} & \multicolumn{5}{c}{{\bf Percentiles}} \\
 & 0.050 & 0.250 & 0.500 & 0.750 & 0.950 \\ \hline
AR(4) & 1.000 & 1.000 & 1.000 & 1.000 & 1.000 \\
DFM-5 & 0.882 & 0.976 & 1.010 & 1.044 & 1.125 \\
CRSIR & 0.652 & 0.759 & 0.865 & 0.991 & 1.186 \\
\bottomrule
\end{tabular}
}

\subfloat[{\normalsize $h=4$}\label{t2c}]{
\begin{tabular}{l|rrrrr}
\toprule
{\bf Method} & \multicolumn{5}{c}{{\bf Percentiles}} \\
 & 0.050 & 0.250 & 0.500 & 0.750 & 0.950 \\ \hline
AR(4) & 1.000 & 1.000 & 1.000 & 1.000 & 1.000 \\
DFM-5 & 0.903 & 0.980 & 1.020 & 1.058 & 1.138 \\
CRSIR & 0.648 & 0.736 & 0.827 & 0.940 & 1.220 \\
\bottomrule
\end{tabular}
}
\end{center}
\end{table}

\begin{table}[htbp]
\begin{adjustwidth}{-1in}{-1in}
\begin{center}
\caption{Median Relative RMSE for Forecasting by Category of Series}
\begin{footnotesize}
\begin{tabular}{l|rrr|rrr|rrr}
\toprule
 & \multicolumn{ 3}{c|}{\textbf{$h=1$}} & \multicolumn{ 3}{c|}{\textbf{$h=2$}} & \multicolumn{ 3}{c}{\textbf{$h=4$}} \\
\multicolumn{1}{c|}{\textbf{Category}} & \multicolumn{1}{c}{\textbf{DFM-5}} & \multicolumn{1}{c}{\textbf{S\&W}} & \multicolumn{1}{c|}{\textbf{CRSIR}} & \multicolumn{1}{c}{\textbf{DFM-5}} & \multicolumn{1}{c}{\textbf{S\&W}} & \multicolumn{1}{c|}{\textbf{CRSIR}} & \multicolumn{1}{c}{\textbf{DFM-5}} & \multicolumn{1}{c}{\textbf{S\&W}} & \multicolumn{1}{c}{\textbf{CRSIR}} \\
\hline
1. GDP Components & 0.905 & 0.905 & 1.079 & 0.907 & 0.870 & 0.807 & 0.906 & 0.906 & 0.839 \\
2. Industrial Production & 0.882 & 0.882 & 0.669 & 0.861 & 0.852 & 0.694 & 0.827 & 0.827 & 0.745 \\
3. Employment & 0.861 & 0.861 & 0.849 & 0.861 & 0.859 & 0.803 & 0.844 & 0.842 & 0.823 \\
4. Unempl. Rate & 0.800 & 0.799 & 0.771 & 0.750 & 0.723 & 0.723 & 0.762 & 0.743 & 0.647 \\
5. Housing  & 0.936 & 0.897 & 1.220 & 0.940 & 0.902 & 1.081 & 0.926 & 0.882 & 0.807 \\
6. Inventories & 0.900 & 0.886 & 0.856 & 0.867 & 0.867 & 0.764 & 0.856 & 0.856 & 0.784 \\
7. Prices & 0.980 & 0.970 & 0.865 & 0.977 & 0.961 & 0.892 & 0.963 & 0.948 & 0.797 \\
8. Wages & 0.993 & 0.938 & 0.967 & 0.999 & 0.919 & 0.960 & 1.019 & 0.931 & 1.031 \\
9. Interest Rates & 0.980 & 0.946 & 0.849 & 0.952 & 0.928 & 0.892 & 0.956 & 0.949 & 0.822 \\
10. Money & 0.953 & 0.926 & 1.000 & 0.933 & 0.921 & 0.950 & 0.909 & 0.909 & 0.927 \\
11. Exchange Rates & 1.015 & 0.981 & 0.974 & 1.015 & 0.980 & 1.108 & 1.036 & 0.965 & 1.150 \\
12. Stock Prices & 0.983 & 0.983 & 0.840 & 0.977 & 0.955 & 0.893 & 0.974 & 0.961 & 1.039 \\
13. Cons. Exp. & 0.977 & 0.977 & 0.765 & 0.963 & 0.960 & 1.082 & 0.966 & 0.955 & 0.963 \\
\bottomrule
\end{tabular}
\end{footnotesize}
\label{t3}
\end{center}
\end{adjustwidth}
\end{table}

\begin{figure}[!tbp]
  \begin{center}
\subfloat[From CRSIR Favored Categories\label{f2a}]{
  \includegraphics[width=0.5\textwidth]{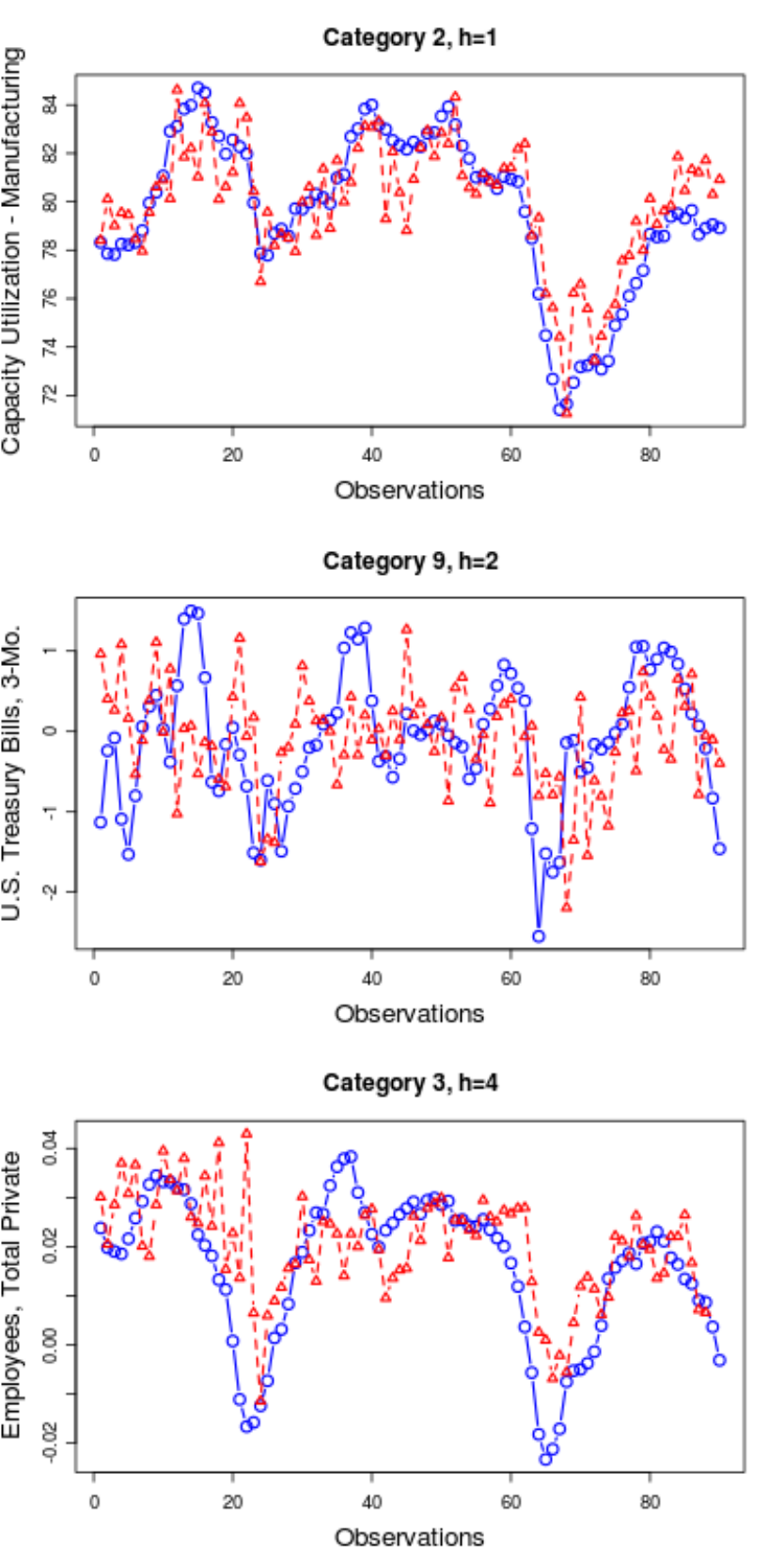}
}
\subfloat[From CRSIR No-Favored Categories\label{f2b}]{
  \includegraphics[width=0.5\textwidth]{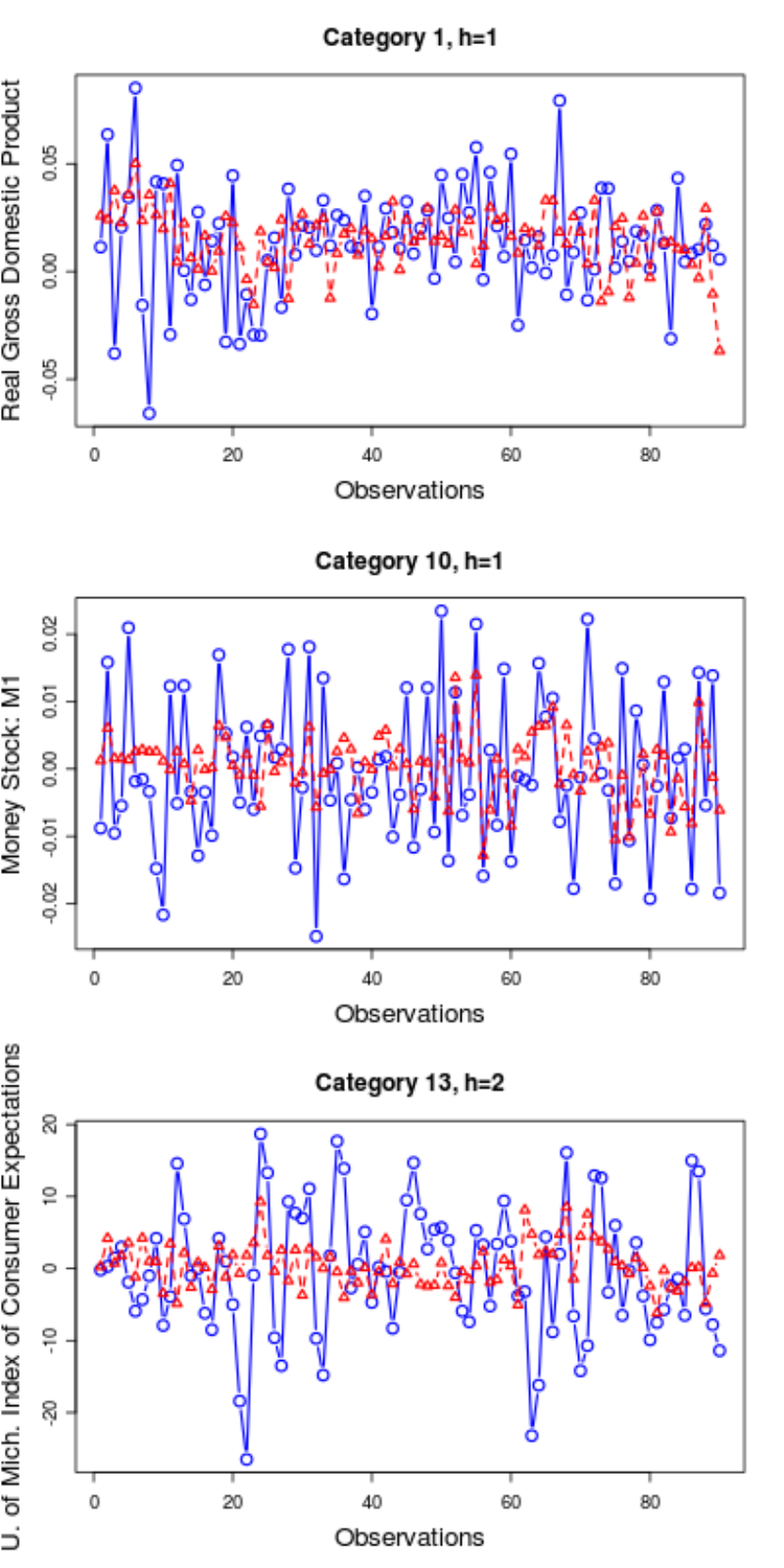}
}
\caption{Plots of the Forecasting Values ({\color{red} $\triangle$}) vs. Real Observations ({\color{blue} $\circ$}) from 1985 to 2008.}
\label{f2}
\end{center}
\end{figure}

Table \ref{t3} presents the median RMSEs relative to AR(4) model by category via cross-validation. Column ``S\&W'' reports the smallest relative RMSE Stock and Watson got using DFM-5 and other shrinkage methods in their 2011 paper. Comparing all these results, CRSIR method has smaller median relative RMSEs for more than 70\% of these categories among three forecasting period, which demonstrates its superiority again.

Table \ref{t3} also indicates the performance of CRSIR varied across categories. It has outstanding performance for some categories, such as \texttt{Industrial Pro\-duction}, \texttt{Unemployment Rate}, \texttt{Inventories}, \texttt{Interest Rates}, etc. But it does not work well for some others, such as \texttt{Housing}, \texttt{Money}, \texttt{Exchange Rates}. Figure \ref{f2} plots six series from both CRSIR favored and no-favored categories. Three of them in Figure \ref{f2a} are from CRSIR favored categories and three of them in Figure \ref{f2b} are from CRSIR no-favored categories. From these plots, one can see that the responses of the CRSIR no-favored series are quite disordered. They are more like white noises, the variations are big but the changes of ${\bf x}$ means are not distinct. The inverse regression method is aimed to detect the variation of $E(\mathbf{x}|y)$. If the conditional expectations of $\mathbf{x}$ do not have much difference for different values of $y$, the estimation for the e.d.r.-directions will be inaccurate, and will lead to the poor performance on forecasting.

Six series are reported as illustrations to show how the estimated RMSE changes when $\tau$ or $c$ changes. They are \texttt{real average hourly earnings} (PI), \texttt{industrial production index} (IP), \texttt{unemployment rate} (UR), \texttt{employees} (EMP), \texttt{3 months treasury bills} (TBI\allowbreak LL), and \texttt{10 years treasury const maturities} (TBOND). Figure \ref{f2} shows the plot of their RMSEs with the values of shrinkage parameter $\tau$ for $h=2$ and $c=10$, Figure \ref{f3} shows the plot of their RMSEs with the number of clusters $c$ for $h=2$ and $\tau=0.5$. RMSEs in both figures are standardizes to the same scale for comparing purposes. These two figures confirm that the clustering and regularized approach do enhance the performance for the regular SIR method, and for this dataset, optimal $\tau$ is between 0.4 to 0.8, and optimal $c$ is between 8 to 12.

\begin{figure}[htpb]
   \begin{center}
   \includegraphics[width=0.85\textwidth]{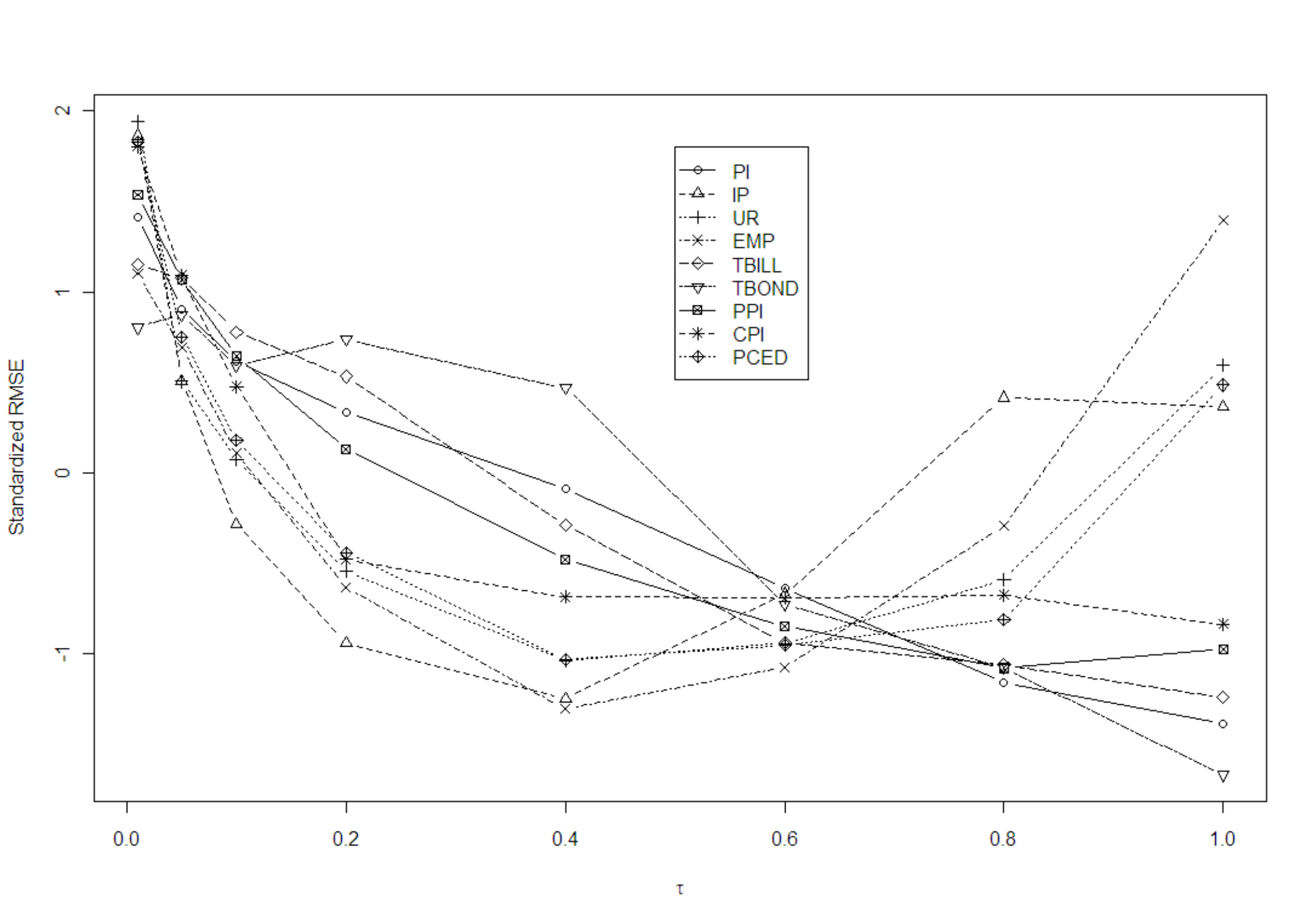}
   \caption{RMSE vs. $\tau$}
   \label{f2}
   \includegraphics[width=0.85\textwidth]{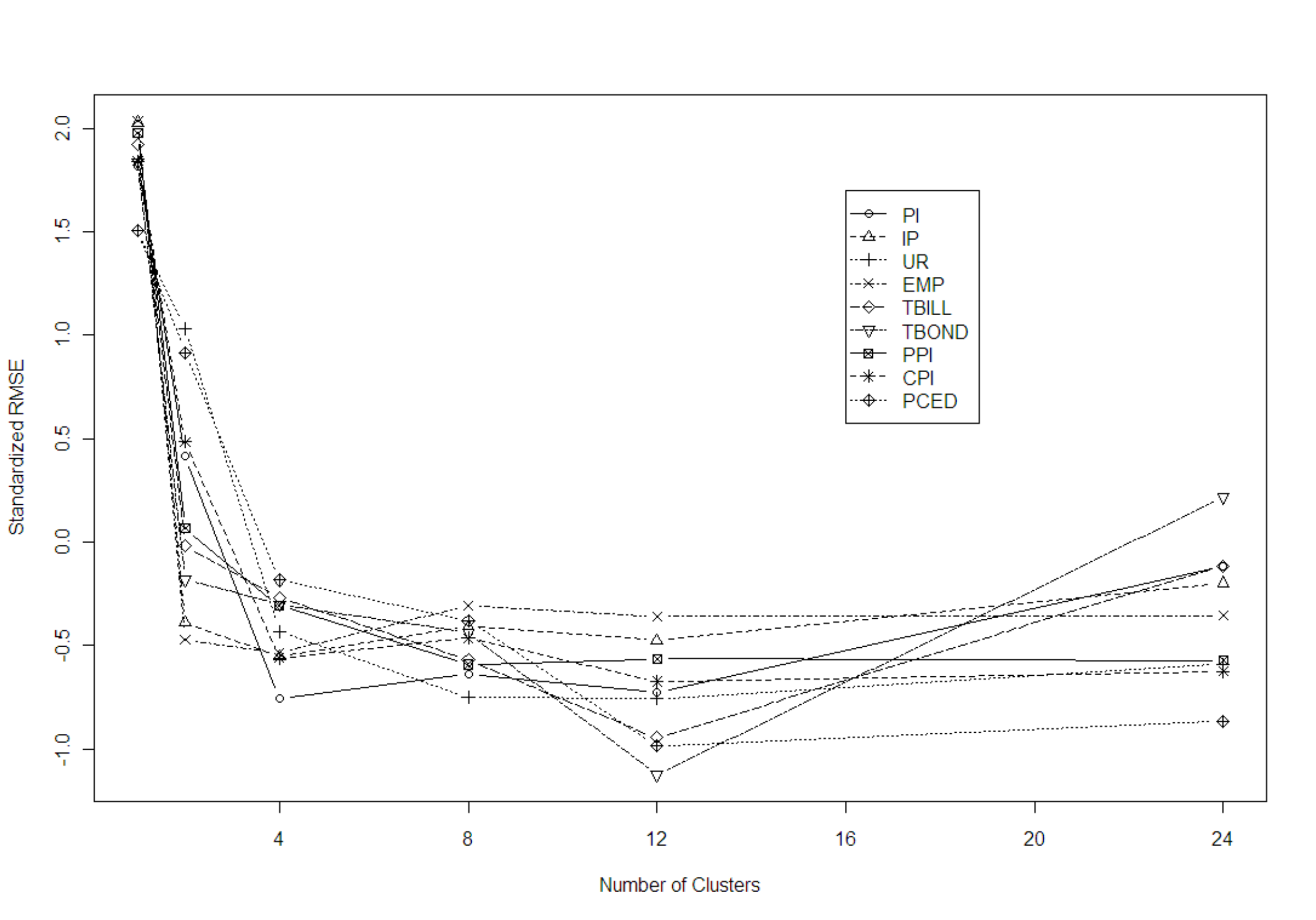}
   \caption{RMSE vs. Number of Clusters}
   \label{f3}
 \end{center}
 \end{figure}

Figure \ref{f4} presents the pair plots of the forecasting series $y$ with the first three e.d.r.-directions estimated by CRSIR for one of the series. The other series in all horizons had similar results. It shows that the relation between the forecasting series and e.d.r.-directions is close to linear, along with the independence among e.d.r.-directions, it is reasonable to use linear model with OLS to predict the values of $y$.

\begin{figure}[tb!]
   \begin{center}
   \includegraphics[width=0.7\textwidth]{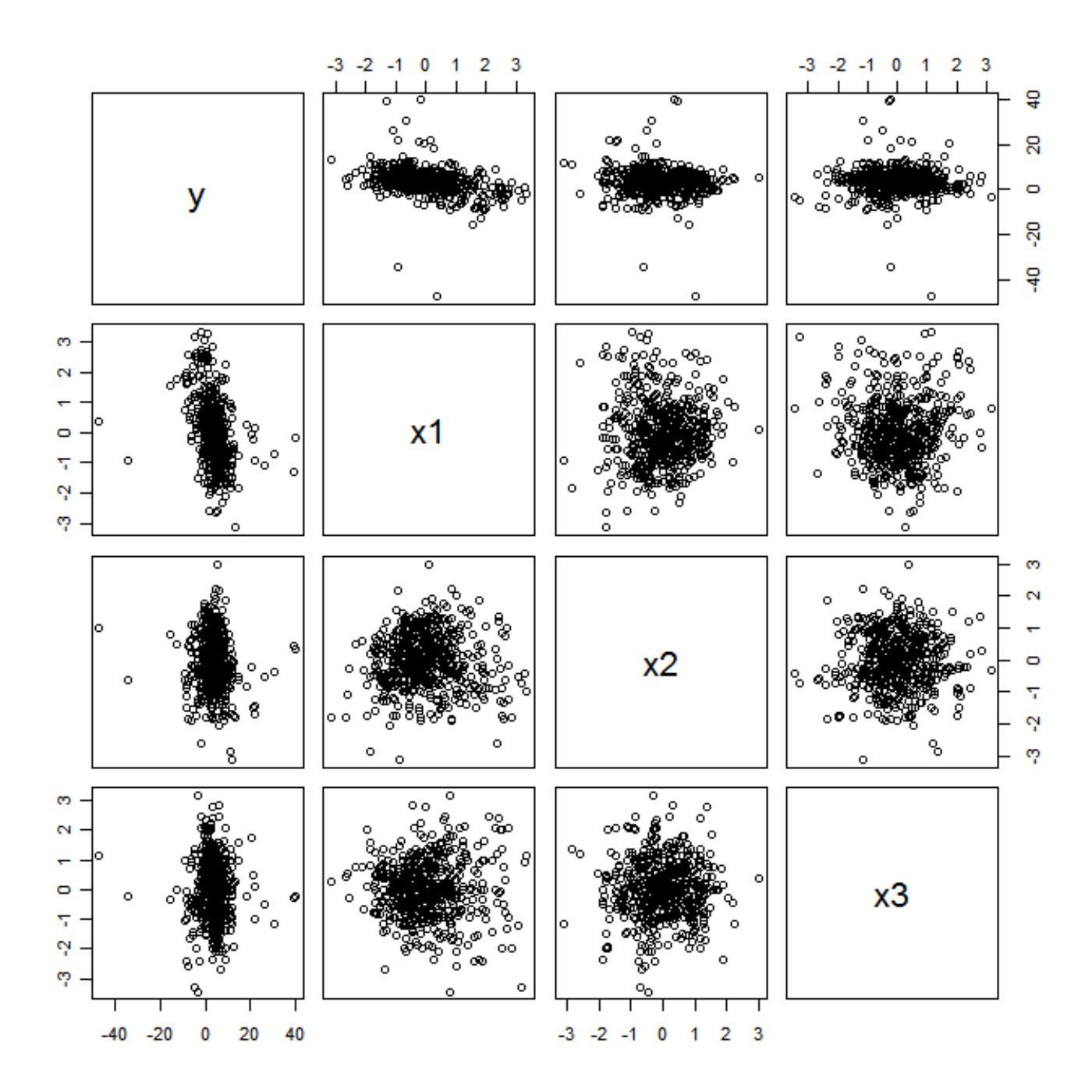}
   \caption{$y$ vs. the First Three e.d.r.-Directions $x1$, $x2$ and $x3$}
 \label{f4}
 \end{center}
\end{figure}

\section{Conclusion and discussion}

\noindent Sliced inverse regression now becomes a popular dimension
reduction method in computer science, engineering and biology. In
this article, we bring it to macroeconomic forecasting model when
there are a large number of predictors and high collinearity.
Compared to the classical dynamic factor model, SIR retrieves
information not only from the predictors but also from the response.
Moreover, our cluster-based regularized SIR has the ability to
handle highly collinearity or ``$T<N$'' data. The simulation
confirms that it offers a lot of improvements over DFM-5 model on
the macroeconomic data set.

After finding the CRSIR variates, we use linear models for forecasting the responses $y$, because scatter plots for CRSIR variates and $y$ values show strong linear relationships, and the results are desirable. But one may use polynomials, splines, Lasso, or some other more advanced regression techniques for different cases to get better fitting results.

Based on its basic idea, there are more than one generalizations of SIR using higher order inverse moments. For instance, SAVE (\citet{cook1991discussion}), SIR-II (\citet{li1991rejoinder}), DR (\citet{li2007directional}), and SIMR (\citet{ye2010sliced}). Our cluster-based algorithm can also be applied to these methods for highly collinearity data, and good performance is expected.

Above all, we can conclude that the cluster-based regularized sliced inverse regression is a powerful tool in forecasting using many predictors. It may not be limited in macroeconomic forecasting, and can also be applied to dimension reduction or variable selection problems in social science, microarray analysis, or clinical trails when the dataset is large and highly correlated.

\bibliographystyle{unsrtnat}
\bibliography{myref}

\section*{Appendix}
Assume the following linearity conditions.
\begin{lcon}\label{con1}
For any $\mathbf{b} \in\mathbb{R}^{N_i}$, the conditional expectation $E(\mathbf{b}'\mathbf{x_i}|\boldsymbol{\theta}_j^{(i)\prime}\mathbf{x_i})$ is linear in $\boldsymbol{\theta}_j^{(i)\prime}\mathbf{x_i}$, $j=1,2,\dots,k_i$.
\end{lcon}
\begin{lcon}\label{con2}
For any $\mathbf{b} \in\mathbb{R}^{N}$, the conditional expectation $E(\mathbf{b}'\mathbf{x}|\Lambda'\mathbf{x})$ is linear in $\boldsymbol{\lambda}_1'\mathbf{x},\boldsymbol{\lambda}_2'\mathbf{x},\dots, \boldsymbol{\lambda}_m'\mathbf{x}$.
\end{lcon}
\begin{lcon}\label{con3}
For any $\mathbf{b} \in\mathbb{R}^{m}$, the conditional expectation $E(\mathbf{b}'\Lambda'\mathbf{x}|\Gamma'\Lambda'\mathbf{x})$ is linear in $\boldsymbol{\gamma}_1'\Lambda'\mathbf{x},\boldsymbol{\gamma}_2'\Lambda'\mathbf{x},\dots,\boldsymbol{\gamma}_v'\Lambda'\mathbf{x}$.
\end{lcon}
Condition \ref{con1} and \ref{con2} are satisfied when all the $\mathbf{x}$'s have elliptical symmetric distribution, especially the multivariate normal distribution (\citet{eaton1986characterization}). Condition \ref{con3} is also satisfied when all the $\Lambda'\mathbf{x}$ have elliptical symmetric distribution, which is true because all the elliptical symmetric distributed $\mathbf{x}$'s have been standardized to the same scale.

Li's Theorem \ref{thm0} can be restated as following for each cluster when $E(\mathbf{x})=0$.
\begin{thm}(\citet{li1991sliced})\label{thm1}
Under Linearity Condition \ref{con1}, $E(\mathbf{x_i}|y)$ is contained in the space spanned by $\Sigma_{\mathbf{x_i}}\boldsymbol{\theta}_j^{(i)},\ j=1,\dots,k_i$
\end{thm}
Furthermore, it's not hard to see that,
\begin{cor}\label{cor1}
Under Linearity Condition \ref{con2}, $E(\mathbf{x}|y)$ is contained in the space spanned by $\Sigma_{\mathbf{x}}\boldsymbol{\lambda}_1,\Sigma_{\mathbf{x}}\boldsymbol{\lambda}_2,\dots,\Sigma_{\mathbf{x}}\boldsymbol{\lambda}_m$.
\end{cor}
\begin{cor}\label{cor2}
Under Linearity Condition \ref{con3}, $E(\Lambda'\mathbf{x}|y)$ is contained in the space spanned by $\Sigma_{\mathbf{\Lambda'x}}\boldsymbol{\gamma}_1,\Sigma_{\mathbf{\Lambda'x}}\boldsymbol{\gamma}_2,\dots,\Sigma_{\mathbf{\Lambda'x}}\boldsymbol{\gamma}_v$.
\end{cor}

Based on the above results, we can conclude that
\begin{thm}\label{thm4}
Under Linearity Conditions \ref{con1}, \ref{con2} and \ref{con3}, $E(\mathbf{x}|y)$ is contained in the space spanned by $\Sigma_{\mathbf{x}}\Lambda \Gamma$.
\end{thm}

\begin{proof}

\citet{li2000high} proved Theorem \ref{thm0}, which is the same as Corollary \ref{cor1} in different notations, by showing that $E(\mathbf{x}|y)$ can be written as,
\begin{equation*}\label{3}
E(\mathbf{x}|y)=\Sigma_{\mathbf{x}}\Lambda\kappa_1(y),
\end{equation*}
where $\kappa_1(y)=(\Lambda'\Sigma_{\mathbf{x}}\Lambda)^{-1}E(\Lambda'\mathbf{x}|y)$.

Similarly, under Condition \ref{con3},
\begin{equation*}
E(\Lambda'\mathbf{x}|y)=\Sigma_{\Lambda'\mathbf{x}}\Gamma\kappa_2(y),
\end{equation*}
where $\kappa_2(y)=(\Gamma'\Sigma_{\Lambda'\mathbf{x}}\Gamma)^{-1}E(\Gamma'\Lambda'\mathbf{x}|y)$ and $\Sigma_{\Lambda'\mathbf{x}}=\Lambda'\Sigma_{\mathbf{x}}\Lambda$.

Therefore,
\begin{eqnarray*}
E(\mathbf{x}|y)=\Sigma_{\mathbf{x}}\Lambda\kappa_1(y)&=&\Sigma_{\mathbf{x}}\Lambda(\Lambda'\Sigma_{\mathbf{x}}\Lambda)^{-1}E(\Lambda'\mathbf{x}|y)\nonumber\\
&=&\Sigma_{\mathbf{x}}\Lambda(\Lambda'\Sigma_{\mathbf{x}}\Lambda)^{-1}\Lambda'\Sigma_{\mathbf{x}}\Lambda \Gamma\kappa_2(y)\nonumber\\
&=&\Sigma_{\mathbf{x}}\Lambda \Gamma\kappa_2(y).
\end{eqnarray*}
That implies that $E(\mathbf{x}|y)$ is in the e.d.r. space spanned by $\Sigma_{\mathbf{x}}\Lambda \Gamma$.

\end{proof}

\end{document}